\documentclass[11pt]{article}
\usepackage[utf8]{inputenc}
\usepackage{geometry}
\usepackage{mathrsfs}
\usepackage{amsmath}
\usepackage{amssymb}
\usepackage{amsthm}

\usepackage{authblk}

\makeatletter

\newcommand{\noun}[1]{\textsc{#1}}

\usepackage{cite}

\usepackage{tikz}

\DeclareMathAlphabet{\mathscr}{OT1}{pzc}{m}{it}

\usepackage[ruled,vlined]{algorithm2e}

\usetikzlibrary{patterns}

\newtheorem{theorem}{Theorem}

\newtheorem{lemma}[theorem]{Lemma}

\newcommand{\opt}{\textsc{Opt}}
\newcommand{\OPT}{\opt}

\newcommand{\I}{\mathcal{I}}

\newcommand{\innerand}{\text{ and }}
\newcommand{\classNP}{\mathsf{NP}}
\newcommand{\classcoNP}{\mathsf{coNP}}
\newcommand{\classPSPACS}{\mathsf{PSPACE}}

\renewcommand{\epsilon}{\varepsilon}

\makeatother

\begin{document}

\title{Packing a Knapsack of Unknown Capacity}

\author{Yann Disser\thanks{Supported by the Alexander von Humboldt Foundation.}}
\author{Max Klimm}
\author{Nicole Megow\thanks{Supported by the German Science Foundation (DFG) under contract  ME 3825/1.}}
\author{Sebastian Stiller}

\affil{\small Department of Mathematics, Technische Universit\"at Berlin, Germany.
\\
\texttt{\{disser,klimm,nmegow,stiller\}@math.tu-berlin.de}.}
\maketitle
\begin{abstract}
We study the problem of packing a knapsack without knowing its capacity.
Whenever we attempt to pack an item that does not fit, the item is
discarded; if the item fits, we have to include it in the packing.
We show that there is always a policy that packs a value within factor
2 of the optimum packing, irrespective of the actual capacity. If
all items have unit density, we achieve a factor equal to the golden
ratio $\varphi\approx1.618$. Both factors are shown to be best possible.

In fact, we obtain the above factors using packing policies that are
\emph{universal} in the sense that they fix a particular order of
the items and try to pack the items in this order, independent of
the observations made while packing. We give efficient algorithms
computing these policies. On the other hand, we show that, for any
$\alpha>1$, the problem of deciding whether a given universal policy
achieves a factor of $\alpha$ is $\classcoNP$-complete. If $\alpha$
is part of the input, the same problem is shown to be $\classcoNP$-complete
for items with unit densities. Finally, we show that it is $\classcoNP$-hard
to decide, for given $\alpha$, whether a set of items admits a universal
policy with factor $\alpha$, even if all items have unit densities.
\end{abstract}
\smallskip
\section{Introduction}

In the standard knapsack problem we are given a set of items, each
associated with a size and a value, and a capacity of the knapsack.
The goal is to find a subset of the items with maximum value who's
size does not exceed the capacity. In this paper, we study the \emph{oblivious
}knapsack problem where the capacity of the knapsack is not given.
Whenever we try to pack an item, we observe whether or not it fits
the knapsack. If it does, the item is packed into the knapsack and
cannot be removed later. If it does not fit, we discard it and continue
packing with the remaining items. The central question of this paper
is how much we loose by not knowing the capacity, in the worst case.
The oblivious variant of the knapsack problem naturally arises whenever
items are prioritized by a different entity or at a different time
than the actual packing of the knapsack.

A solution to the oblivious knapsack problem is a policy that governs
the order in which we attempt to pack the items, depending only on
the observation which of the previously attempted items did fit into
the knapsack and which did not. In other words, a policy is a binary
decision tree with the item that is tried first at its root. The two
children of the root are the items that are tried next, which of the
two depends on whether or not the first item fits the knapsack, and
so on. We aim for a solution that is good for \emph{every} possible
capacity, compared to the best solution of the standard knapsack problem
for this capacity. Formally, a policy has \emph{robustness factor
$\alpha$ }if, for any capacity, packing according to the policy results
in a value that is at least a $1/\alpha$-fraction of the optimum
value for this capacity.

We show that the oblivious knapsack problem always admits a robustness
factor of 2. In fact, this robustness factor can be achieved with
a policy that packs the items according to a fixed order, irrespective
of the observations made while packing. Such a policy is called \emph{universal.
}We provide an algorithm that computes a 2-robust, universal policy
in time $\Theta(n\log n)$ for a given set of $n$ items. We complement
this result by showing that no robustness factor better than 2 can
be achieved in general, even by policies that are not universal. In
other words, the cost of not knowing the capacity is exactly 2.

We give a different efficient algorithm for the case that all items
have unit density, i.e., size and value of each item coincide. This
algorithm produces a universal policy with a robustness factor of
at most the golden ratio $\varphi\approx1.618$. Again, we show that
no better robustness factor can be achieved in general, even by policies
that are not universal. 

While good universal policies can be found efficiently, it is intractable
to compute the robustness factor of a \emph{given} universal policy
and it is intractable to compute the best robustness factor an instance
admits. Specifically, we show that, for any\emph{ }fixed\emph{ }$\alpha\in(1,\infty)$,
it is $\classcoNP$-complete to decide whether a given universal policy
is $\alpha$-robust. For unit densities we establish a slightly weaker
hardness result by showing that it is $\classcoNP$-complete to decide
whether a given universal policy achieves a \emph{given} robustness
factor $\alpha$. Finally, we show that, for given $\alpha$, it is
$\classcoNP$-hard to decide whether an instance of the oblivious
knapsack problem admits a universal policy with robustness factor
$\alpha$, even when all items have unit density.

\subsection*{Related work}

The knapsack problem has been studied for different models of imperfect
information. In the \emph{stochastic} knapsack problem, sizes and
values of the items are random variables. It is known that a policy
maximizing the expected value is $\classPSPACS$-hard to compute,
see Dean et al.~\cite{deanGV04}. The authors assume that the packing
stops when the first item does not fit the knapsack, and give a universal
policy that approximates the value obtained by an optimal, not necessarily
universal, policy by a factor of~$2$. They also provide a non-universal
policy within a factor of $3+\epsilon$ of the optimal policy. Bhalgat
et al.~\cite{bhalgatGK11} give an algorithm with an improved approximation
guarantee of $8/3+\epsilon$. They also give a PTAS for the case that
it is allowed to violate the capacity of the knapsack by a factor
of $1+\epsilon$.

In \emph{robust} knapsack problems, a set of possible scenarios for
the sizes and values of the items is given. Yu~\cite{yu96}, Bertsimas
and Sim~\cite{bertsimasS03}, Goetzmann et al.~\cite{GoetzmannST11},
and Monaci and Pferschy~\cite{monaciP11} study the problem of maximizing
the worst-case value of a knapsack under various models. Büsing et
al.~\cite{busingKK11-discrete} and Bouman et al.~\cite{boumanAH11}
study the problem from a computational point of view. Both allow for
an adjustment of the solution after the realization of the scenario.
Similar to our model, Bouman et al.~consider uncertainty in the capacity.

The notion of a \emph{robustness factor} that we adopt in this work
is due to Hassin and Rubinstein~\cite{hassinR02} and is defined
as the worst-case ratio of solution and optimum, over all realizations.
Kakimura et al.~\cite{kakimuraMS11} analyze the complexity of deciding
whether an $\alpha$-robust solution exists for a knapsack instance
with an unknown bound on the number of items that can be packed. Megow
and Mestre~\cite{megowM13} study a variant of the knapsack problem
with unknown capacity closely related to ours. In contrast to our
model, they assume that the packing stops once the first item does
not fit the remaining capacity. In this model, a universal policy
with a constant robustness factor may fail to exist, and, thus, Megow
and Mestre resort to \emph{instance-sensitive} performance guarantees.
They provide a PTAS that constructs a universal policy with robustness
factor arbitrarily close to the best possible robustness factor for
every particular instance.

The concept of \emph{obliviousness }is used in various other contexts
(explicitly or implicitly), such as hashing~(Carter and Wegman~\cite{carterW79}),
caching (Frigo et al.~\cite{frigoLPR99}, Bender et al.~\cite{benderCD02})
routing~(Valiant and Brebner~\cite{valiantB81}, Räcke~\cite{raecke09}),
TSP (Papadimitriou~\cite{papadimitriou94}, Deineko et al.~\cite{deinekoRW95},
Jia et al,~\cite{jiaLNRS05}), Steiner tree and set cover~(Jia et
al,~\cite{jiaLNRS05}), and scheduling~(Epstein et al.~\cite{epsteinLMMMSS12},
Megow and Mestre~\cite{megowM13}). In all of these works, the general
idea is that specific parameters of a problem instance are unknown,
e.g., the cache size or the set of vertices to visit in a TSP tour,
and the goal is to find a \emph{universal solution }that performs
well for all realizations of the hidden parameters.

Universal policies for the oblivious knapsack problem play a role
in the design of public key cryptosystems. One of the first such systems
-- the Merkle–Hellman knapsack cryptosystem \cite{merkleH78} -- is
based on particular instances that allow for a $1$-robust universal
policy for the oblivious knapsack problem. The basic version of this
cryptosystem can be attacked efficiently, e.g.,~by the famous attack
of Shamir~\cite{shamir82}. This attack uses the fact that the underlying
knapsack instance has exponentially increasing item sizes. A better
understanding of universal policies may help to develop knapsack-based
cryptosystems that avoid the weaknesses of Merkle and Hellman's.

\section{Preliminaries}

An instance of the \emph{oblivious knapsack problem} is given by a
set of $n$ items~$\mathcal{I}$, where each item $i\in\mathcal{I}$
has a non-negative \emph{value} $v(i)\in\mathbb{Q}_{\geq0}$ and a
strictly positive \emph{size} $l(i)\in\mathbb{Q}_{>0}$. For a subset
$S\subseteq\I$ of items, we write $v(S)=\sum_{i\in S}v(i)$ and $l(S)=\sum_{i\in S}l(i)$
to denote its total value and total size, respectively, of the items
in $S$. A \emph{solution} for instance $\mathcal{I}$ is a policy
$\mathscr{P}$ that governs the order in which the items are considered
for packing into the knapsack. The policy must be independent of the
capacity of the knapsack, but the choice which item to try next may
depend on the observations which items did and which items did not
fit the knapsack so far. Formally, a solution policy is a binary decision
tree that contains every item exactly once along each path from the
root to a leaf. The \emph{packing} $\mathscr{P}(C)\subseteq\I$ of
$\mathscr{P}$ for a fixed capacity $C$ is obtained as follows: We
start with $\mathscr{P}(C)=\emptyset$ and check whether the item
$r$ at the root of $\mathscr{P}$ fits the knapsack, i.e., whether
$l(r)+l(\mathscr{P}(C))\leq C$. If the item fits, we add $r$ to
$ $$\mathscr{P}(C)$ and continue packing recursively with the left
subtree of $r$. Otherwise, we discard $r$ and continue packing recursively
with the right subtree of $r$.

A \emph{universal policy }$\Pi$ for instance $\mathcal{I}$ is a
policy that does not depend on observations made while packing, i.e.,
the decision tree for a universal policy has a fixed permutation of
the items along every path from the root to a leaf. We identify a
universal policy with this fixed permutation and write $\Pi=(\Pi_{1},\Pi_{2},\dots,\Pi_{n})$.
Analogously to general policies, the packing $\Pi(C)\subseteq\mathcal{I}$
of a universal policy~$\Pi$ for capacity $C\leq l(\mathcal{I})$
is obtained by considering the items in the order given by the permutation~$\Pi$
and adding every item if it does not exceed the remaining capacity. 

We measure the quality of a policy for the oblivious knapsack problem
by comparing its packing with the optimal packing for each capacity.
More precisely, a policy $\mathscr{P}$ for instance $\mathcal{I}$
is called~$\alpha$-\emph{robust} \emph{for capacity} $C$, $\alpha\geq1$,
if it holds that~$v(\textsc{Opt}(\mathcal{I},C))\leq\alpha\cdot v(\mathscr{P}(C))$,
where $\textsc{Opt}(\mathcal{I},C)$ denotes an optimal packing for
capacity $C$. We say $\mathscr{P}$ is $\alpha$-\emph{robust} if
it is $\alpha$-robust for all capacities. In this case, we call~$\alpha$
the \emph{robustness factor} of policy $\mathscr{P}$.

\section{Solving the Oblivious Knapsack Problem\label{sec:Algorithm}}

In this section, we describe an efficient algorithm that constructs
a universal policy for a given instance of the oblivious knapsack
problem. The solution produced by our algorithm is guaranteed to pack
at least half the value of the optimal solution for any capacity $C$.
We show that this is the best possible robustness factor. 

The analysis of our algorithm relies on the classical \emph{modified
greedy} algorithm (cf.~\cite{korteV02}). We compare the packing
of our policy, for each capacity, to the packing obtained by the modified
greedy algorithm instead of the actual optimum. As the modified greedy
is a 2-approximation, to show that our policy is 2-robust it is sufficient
to show that its packing is never worse the one obtained by the modified
greedy algorithm. We briefly review the modified greedy algorithm.

\begin{algorithm}[b]   
  \DontPrintSemicolon 
  \KwIn{set of items $\I$, capacity $C$}
  \KwOut{subset $S \subseteq \I$ such that $l(S) \leq C$ and $v(S) \geq v(\OPT(\I,C))/2$}     
  \caption{\textsc{MGreedy}($\mathcal{I}$, $C$)}\label{alg:mgreedy}     

  $D \leftarrow \left<\text{items in }\{i\in\mathcal{I} \mid l(i) \leq C\}\text{ sorted  decreasingly by density} \right>$\;
  $k \;\leftarrow \max \{j \mid l(\{D_1,\dots,D_j\}) \leq C\}$\;   
  $P \leftarrow (D_1,\dots,D_{k})$, $s \leftarrow D_{k+1}$\;   
  \uIf{$v(P) \geq v(s)$}   
  {\Return $P$\; }   
  \Else   {\Return $\{s\}$\; } 
\end{algorithm}

Let~$d(i)=v(i)/l(i)$ denote the \emph{density} of item $i$. The
modified greedy algorithm (\textsc{MGreedy}) for a set of items~$\mathcal{I}$
and known knapsack capacity~$C$ first discards all items that are
larger than~$C$ from $\mathcal{I}$. The remaining items are sorted
in non-increasing order of their densities, breaking ties arbitrarily.
The algorithm then either takes the longest prefix~$P$ of the resulting
sequence that still fits into capacity~$C$, or the first item~$s$
that does not fit anymore, depending on which of the two has a greater
value. In the latter case, we say that~$s$ is a \emph{swap item}
(for capacity $C$) that and $C$ is a \emph{swap capacity}. In both
cases, we refer to~$P$ as the \emph{greedy set} for capacity~$C$.
See Algorithm~\ref{alg:mgreedy} for a formal description. 

For our analysis, it is helpful to fix the tie-breaking rule of the
greedy algorithm. To this end, we assume that there is a bijection
$t:\I\to\{1,2,\dots,n\}$, that maps every item $i\in\I$ to a \emph{tie-breaking
index} $t(i)$, and that the modified greedy algorithm initially sorts
the items decreasingly with respect to the tuple $\tilde{d}(\cdot)=(d(\cdot),t(\cdot))$,
i.e., the items are sorted non-increasingly by density and whenever
two items have the same density, they are sorted by decreasing tie-breaking
index. In the following, for two items $i,j$, we write $\tilde{d}(i)\succ\tilde{d}(j)$
if and only if $d(i)>d(j)$, or $d(i)=d(j)$ and $t(i)>t(j)$, and
say that $i$ has higher density than $j$.

We evaluate the quality of our universal policy by comparing it for
every capacity with the solution of \textsc{MGreedy}. This analysis
suffices because of the following well-known property of the modified
greedy algorithm.
\begin{theorem}
[cf.~\cite{korteV02}] For every instance $(\mathcal{I},C)$ of the
standard knapsack problem with known capacity, $v(\textsc{Opt}(\mathcal{I},C))\leq2\cdot v(\textsc{MGreedy}(\mathcal{I},C))$.\label{thm:greedy}

\end{theorem}
We are now ready to describe our algorithm $\textsc{Universal}$~(Algorithm~\ref{our_algo})
that produces a universal policy tailored to imitate the behavior
of $\textsc{MGreedy}$ without knowing the capacity. 

First, $\textsc{Universal}$ determines which items are swap items.
It then starts with an empty permutation, and considers the items
in order of non-decreasing sizes, inserting each item into the permutation.
Swap items are always placed in front of all items already in the
permutation, and all other items are inserted in front of the first
item in the permutation that has a lower density.

We prove the following result.

\RestyleAlgo{ruled,vlined} 
\begin{algorithm}[tb]   
  \DontPrintSemicolon 
  \KwIn{set of items $\I$} 
  \KwOut{sequence of items $\Pi$}
  \caption{\textsc{Universal}($\mathcal{I}$)\label{our_algo}}
 
  $L \leftarrow \left<\text{items in }\mathcal{I}\text{ sorted by non-decreasing size}\right>$\;
  $\Pi^{(0)} \leftarrow \emptyset$\;
  
  \For{$r \leftarrow 1,\dots,n$}   
  {     
    \uIf{$L_r$ is a swap item}     
    {$\Pi^{(r)} \leftarrow (L_r,\Pi^{(r-1)})$\;}     
    \Else     
    {       
      $j \leftarrow 1$\;       
      \While{$j \leq |\Pi|$ {\bf and} $\tilde{d}(\Pi_j) \succ \tilde{d}(L_r)$}       
      { $j \leftarrow j+1$\; }       
      $\Pi^{(r)} \leftarrow (\Pi^{(r-1)}_1,\dots,\Pi^{(r-1)}_{j-1},L_r,\Pi^{(r-1)}_j,\dots)$\;     
    }   
  }
  \Return $\Pi^{(n)}$\; 
\end{algorithm}
\begin{theorem}
The algorithm \noun{Universal} constructs a universal policy of robustness
factor~$2$. \label{thm:2-competitive} 
\end{theorem}
Before we prove this theorem, we first analyze the structure of the
permutation output by $\textsc{Universal}$ in terms of density, size,
and value. First, we prove that every item following a non-swap item
has lower density.
\begin{lemma}
For a sequence $\Pi$ returned by $\textsc{Universal}$, we have $\tilde{d}(\Pi_{k})\succ\tilde{d}(\Pi_{k+1})$
for every non-swap item $\Pi_{k}$, $1\leq k<n$. \label{lem:densities} \end{lemma}
\begin{proof}
For $j\in\{k,k+1\}$, let $r(j)\in\{1,\dots,n\}$ be the index of
the iteration in which $\textsc{Universal}$ inserts $\Pi_{j}$ into
$\Pi$. We distinguish two cases.

If $r(k)<r(k+1)$, then the item $\Pi_{k+1}$ cannot be a swap item,
since it would appear in front of the item $\Pi_{k}$ if it was. As
each non-swap item is inserted into $\Pi$ such that all items left
of it are larger with respect to $\tilde{d}$, the claim follows.

If $r(k)>r(k+1)$, since it is not a swap item, $\Pi_{k}$ is put
in front of $\Pi_{k+1}$ because it has a higher density.  
\end{proof}
We prove that no item preceding a swap item has smaller size.
\begin{lemma}
For a permutation $\Pi$ returned by $\textsc{Universal}$, we have
$l(\Pi_{j})\geq l(\Pi_{k})$ for every swap item $\Pi_{k},1<k\leq n$,
and every other item $\Pi_{j},1\leq j<k$. \label{lem:sizes} \end{lemma}
\begin{proof}
Since $\Pi_{k}$ is a swap item, it stands in front of all items inserted
earlier into $\Pi$. Hence, all items that appear in front of $\Pi_{k}$
in $\Pi$ have been inserted in a later iteration of $\textsc{Universal}$.
Since $\textsc{Universal}$ processes items in order of non-decreasing
sizes, we have $l(\Pi_{j})\geq l(\Pi_{k})$.  
\end{proof}
We prove that no item preceding a swap item has smaller value.
\begin{lemma}
For a permutation $\Pi$ returned by $\textsc{Universal}$, we have
$v(\Pi_{j})\geq v(\Pi_{k})$ for every swap item $\Pi_{k},1<k\leq n$,
and every other item $\Pi_{j},1\leq j<k$. \label{lem:values} \end{lemma}
\begin{proof}
We distinguish three cases.

\emph{First case:} $\Pi_{j}$ is a swap item and $d(\Pi_{j})\geq d(\Pi_{k})$.
By Lemma~\ref{lem:sizes}, we have $l(\Pi_{j})\geq l(\Pi_{k})$,
and the claim trivially holds.

\emph{Second case:} $\Pi_{j}$ is a swap item and $d(\Pi_{j})<d(\Pi_{k})$.
Since $\Pi_{j}$ is a swap item, there is a capacity $C\geq l(\Pi_{j})$
such that
\[
v(\Pi_{j})>v(\{i\in\I\mid l(i)\leq C\innerand\tilde{d}(i)\succ\tilde{d}(\Pi_{j})\}).
\]
In particular, for $C=l(\Pi_{j})$ we obtain
\begin{equation}
v(\Pi_{j})>v(\{i\in\mathcal{I}\mid l(i)\leq l(\Pi_{j})\innerand\tilde{d}(i)\succ\tilde{d}(\Pi_{j})\}).\label{eq:rhs}
\end{equation}

Since, by Lemma~\ref{lem:sizes}, $l(\Pi_{j})\geq l(\Pi_{k})$, the
item $\Pi_{k}$ is included in the set on the right hand side of~\eqref{eq:rhs}.
We conclude that $v(\Pi_{j})>v(\Pi_{k})$.

\emph{Third case:} $\Pi_{j}$ is not a swap item. Let $\Pi_{j'}$
be the first swap item after $\Pi_{j}$ in $\Pi$, i.e., 
\begin{align*}
j'=\min\{i\in\{j+1,\dots,k\}\mid\Pi_{i}\text{ is a swap item }\}.
\end{align*}
Note that the minimum is attained as $\Pi_{k}$ is a swap item. The
analysis of the first two cases implies that $v(\Pi_{j'})\geq v(\Pi_{k})$.
By Lemma~\ref{lem:densities} we have $d(\Pi_{j})\geq d(\Pi_{j+1})\geq\dots\geq d(\Pi_{j'})$,
and by Lemma~\ref{lem:sizes} we have $l(\Pi_{j})\geq l(\Pi_{j'})$.
Hence, $v(\Pi_{j})\geq v(\Pi_{j'})\geq v(\Pi_{k})$.  
\end{proof}
Finally, the next lemma gives a legitimation for the violation of
the density order in the output permutation. Essentially, whenever
an item precedes denser items, we guarantee that it is worth at least
as much as all of them combined.
\begin{lemma}
For a permutation $\Pi$ returned by $\textsc{Universal}$, we have
\[
v(\Pi_{k})\geq v\bigl(\bigl\{\Pi_{j}\mid j>k\innerand\tilde{d}(\Pi_{j})\succ\tilde{d}(\Pi_{k})\bigr\}\bigr)
\]
for every item $\Pi_{k},1\leq k<n$. \label{lem:order_violation_legitimation} \end{lemma}
\begin{proof}
We distinguish whether $\Pi_{k}$ is a swap item, or not.

If $\Pi_{k}$ is a swap item, by definition, $\Pi_{k}$ is worth more
than the greedy set for some capacity $C\geq l(\Pi_{k})$. Thus,
\[
v(\Pi_{k})>v\bigl(\bigl\{\Pi_{j}\mid l(\Pi_{j})\leq C\innerand\tilde{d}(\Pi_{j})\succ\tilde{d}(\Pi_{k})\bigr\}\bigr)\geq
v\bigl(\bigl\{\Pi_{j}\mid l(\Pi_{j})\leq l(\Pi_{k})\innerand\tilde{d}(\Pi_{j})\succ\tilde{d}(\Pi_{k})\bigr\}\bigr).
\]
 Since items whose size is strictly larger than $l(\Pi_{k})$ are
inserted into $\Pi$ at a later iteration of $\textsc{Universal}$,
they can only end up behind $\Pi_{k}$ if they are smaller with respect
to $\tilde{d}$. Hence, 
\begin{align*}
 & \{\Pi_{j}\,|\, j>k\innerand\tilde{d}(\Pi_{j})\succ\tilde{d}(\Pi_{k})\}\subseteq\{\Pi_{j}\mid l(\Pi_{j})\leq
l(\Pi_{k})\innerand\tilde{d}(\Pi_{j})\succ\tilde{d}(\Pi_{k})\},
\end{align*}
and thus $v(\Pi_{k})>v(\{\Pi_{j}\,|\, j>k\innerand\tilde{d}(\Pi_{j})\succ\tilde{d}(\Pi_{k})\})$,
as claimed.

If, on the other hand, $\Pi_{k}$ is not a swap item, let $\Pi_{k'}$
be the first swap item after it in $\Pi$. If no such item exists,
the claim holds by Lemma~\ref{lem:densities}, since 
\begin{align*}
\bigl\{\Pi_{j}\mid j>k\innerand\tilde{d}(\Pi_{j})\succ\tilde{d}(\Pi_{k})\bigr\}=\emptyset.
\end{align*}
Otherwise, by Lemma~\ref{lem:densities}, we obtain
$\tilde{d}(\Pi_{k})\succ\tilde{d}(\Pi_{k+1})\succ\dots\succ\tilde{d}(\Pi_{k'})$
and hence 
\begin{align*}
\{\Pi_{j}\mid j>k\innerand\tilde{d}(\Pi_{j})\succ\tilde{d}(\Pi_{k})\} & =\{\Pi_{j}\mid
j>k'\innerand\tilde{d}(\Pi_{j})\succ\tilde{d}(\Pi_{k})\}\\
 & \subseteq\{\Pi_{j}\mid j>k'\innerand\tilde{d}(\Pi_{j})\succ\tilde{d}(\Pi_{k'})\}.
\end{align*}
Consequently, and by the argument above for swap items, 
\begin{align*}
v(\Pi_{k'}) & >v(\{\Pi_{j}\mid j>k'\innerand\tilde{d}(\Pi_{j})\succ\tilde{d}(\Pi_{k'})\})\\
 & \geq v(\{\Pi_{j}\mid j>k\innerand\tilde{d}(\Pi_{j})>\tilde{d}(\Pi_{k})\})).
\end{align*}
Finally, by Lemma~\ref{lem:values}, we have $v(\Pi_{k})\geq v(\Pi_{k'})\geq v(\{\Pi_{j}\,|\,
j>k\innerand\tilde{d}(\Pi_{j})\succ\tilde{d}(\Pi_{k})\})$.
\end{proof}

We now prove Theorem~\ref{thm:2-competitive}.
\begin{proof}
[of Theorem \ref{thm:2-competitive}] We show that for every set of
items $\mathcal{I}$, the permutation $\Pi=\textsc{Universal}(\mathcal{I})$
satisfies $v(\textsc{Opt}(\mathcal{I},C))\leq2v(\Pi(C))$ for every
capacity $C\leq l(\mathcal{I})$. By Theorem~\ref{thm:greedy}, it
suffices to show $v(\Pi(C))\geq v(\textsc{MGreedy}(\mathcal{I},C))$
for all capacities. We distinguish between swap capacities and capacities
where $\textsc{MGreedy}$ outputs a greedy set.

First, assume that $C$ is a swap capacity, and let $\{\Pi_{k}\}=\textsc{MGreedy}(\mathcal{I},C)$
be the swap item returned by the modified greedy algorithm. Then,
$\Pi(C)$ contains at least one item $\Pi_{j}$ with $j\leq k$. By
Lemma~\ref{lem:values} we have 
\begin{align*}
v(\Pi(C))\geq v(\Pi_{j})\geq v(\Pi_{k})=v(\textsc{MGreedy}(\mathcal{I},C)).
\end{align*}

Now assume that $C$ is not a swap capacity. Let $G^{+}=\textsc{MGreedy}(\mathcal{I},C)\setminus\Pi(C)$
be the set of items in the greedy set for capacity $C$ that are not
packed by the permutation $\Pi$. Similarly, let $U^{+}=\Pi(C)\setminus\textsc{MGreedy}(\mathcal{I},C)$.
If $G^{+}=\emptyset$, then $v(\Pi(C))\geq v(\textsc{MGreedy}(\mathcal{I},C))$
and we are done. Suppose now that $G^{+}\neq\emptyset$. Then, also
$U^{+}\neq\emptyset$. For all items $i\in U^{+}$, we have $l(i)\leq C$
and $i\notin\textsc{MGreedy}(\mathcal{I},C)$. Since $C$ is not a
swap capacity, $\textsc{MGreedy}(\mathcal{I},C)$ is the greedy set
for capacity $C$, and thus $\tilde{d}(i)\prec\tilde{d}(i')$ for
all $i\in U^{+}$ and $i'\in G^{+}$. By definition of $\Pi(C)$ and
since $U^{+}\neq\emptyset$, we also have $k=\min\{j\mid\Pi_{j}\in U^{+}\}<\min\{k'\mid\Pi_{k'}\in G^{+}\}$,
i.e., the first item $\Pi_{k}\in U^{+}$ in~$\Pi$ is encountered
before every item from $G^{+}$. It follows that 
\begin{align*}
G^{+}\subseteq\bigl\{\Pi_{j}\mid j>k\innerand\tilde{d}(\Pi_{j})\succ\tilde{d}(\Pi_{k})\bigr)\bigr\}.
\end{align*}
 Using $v(U^{+})\geq v(\Pi_{k})$ and $v(\Pi_{k})\geq v(G^{+})$ (Lemma~\ref{lem:order_violation_legitimation})
we get 
\begin{align*}
v(\Pi(C)) & =v\bigl(\Pi(C)\cap\textsc{MGreedy}(\mathcal{I},C)\bigr)+v(U^{+})\\
 & \geq v\bigl(\Pi(C)\cap\textsc{MGreedy}(\mathcal{I},C)\bigr)+v(G^{+})=v(\textsc{MGreedy}(\mathcal{I},C)).
\end{align*}
\end{proof}
While it is obvious that \noun{Universal} runs in polynomial time,
we show that it can be modified to run in time $\Theta(n\log n)$. 
\begin{theorem}
The algorithm \noun{Universal} can be implemented to run in time $\Theta(n\log n)$.\label{thm:running time
Universal}\end{theorem}
\begin{proof}
We first argue how all swap items can be determined in time $\Theta(n\log n)$.
We use that an item is a swap item if and only if it is worth more
than all smaller items of higher density combined. This is true, because
every item $i$ that is worth more than all smaller items of higher
density is a swap item for capacity $l(i)$. Conversely, a swap item
$i$ for capacity $C\geq l(i)$ is worth more than all items of higher
density that are smaller than $C$.

We maintain a balanced search tree for items that is ordered by size
and stores the total value of the items of both subtrees in the corresponding
root. Inserting an item into this tree as well as determining whether
an item is worth more than all smaller items in the tree both takes
time $\Theta(\log n)$. To determine the set of swap items, we iterate
over all items in order of decreasing densities and insert items one
by one into the search tree. After each insertion, we query whether
the newly inserted item is worth more than all smaller items in the
tree. This is true if and only if the item is worth more than all
smaller items of higher density, i.e., if and only if the item is
a swap item. Including the initial sorting by density, we can determine
all swap items in time $\Theta(n\log n)$.

We construct the output permutation $\Pi$ by iterating over the items
in order of increasing size, as in Algorithm~\ref{our_algo}. We
maintain a list $L$ of balanced search trees, each ordered by density.
Except for the last tree in $L$, every tree contains exactly one
swap item, which is the item of smallest density in the tree. The
density of a tree is the density of this swap item (or 0 if the tree
has no swap item). Each tree stores the items in $\Pi$ to the left
of the corresponding swap item (if it exists) and to the right of
the swap item of the preceding tree in $L$ (if it exists). We start
with a list containing a single tree with no corresponding swap item,
which eventually holds all non-swap items that end up behind the last
swap item in $\Pi$. Whenever we encounter a new swap item, we add
a new tree consisting of only this swap item to the front of $L$.
For each non-swap item, we have to find the correct tree to insert
it into. Once we know the tree, we can determine the position at which
to insert the item into the tree, and thus in $\Pi$, in time $\Theta(\log n)$
simply by searching the tree. 

To complete the proof, we need an efficient way to find the correct
tree in $L$ for a non-swap item. For this purpose, we maintain a
sublist $L'$ of $L$ that contains only those trees that are needed
for the remainder of the algorithm. Whenever a new swap item $s$
adds a tree to the front of $L$, we also add the tree to the front
of~$L'$. Observe that from this point on no items are inserted into
trees of a higher density than $s$. Hence, before inserting the tree
of $s$ to $L'$, we may remove trees of higher density from the front
of~$L'$. This guarantees that $L'$ remains sorted by density. We
can thus implement $L'$ as a balanced search tree order by density.
This way, we can find the correct tree for each non-swap item in time
$\Theta(\log n)$. Since every tree is removed at most once from $L'$,
the amortized cost for maintaining the sublist is constant for each
swap item.

Since $\textsc{Universal}$ requires $n$ iterations, the total running
time is $\Theta(n\log n)$.  
\end{proof}
We now give a general lower bound on the robustness factor of any
policy for the oblivious knapsack problem. This shows that $\textsc{Universal}$
is best possible.
\begin{theorem}
For every $\delta>0$, there are instances of the oblivious knapsack
problem where no policy achieves a robustness factor of $2-\delta$.\end{theorem}
\begin{proof}
We give a family of instances, one for each size $n\geq3$. We ensure
that for every item $i$ of the instance of size $n$, there is a
capacity $C$, such that packing item $i$ first can only lead to
a solution that is worse than $\textsc{Opt}(\mathcal{I},C)$ by a
factor of at least $(2-4/n)$. This completes the proof, as the factor
approaches $2$ for increasing values of $n$.

The instance of size $n$ is given by $\mathcal{I}=\{1,2,\dots,n\}$
with
\begin{align*}
l(i) & =F_{n}+F_{i}-1, & v(i) & =1+\frac{i}{n},
\end{align*}
 where $F_{i}$ denotes the $i$-th Fibonacci number ($F_{1}=1,F_{2}=1,F_{3}=2,\dots$).

We need to show that, no matter which item is tried first (i.e., no
matter which item is the root of the policy), there is a capacity
for which this choice ruins the solution. Observe that both values
and sizes of the items are strictly increasing. Assume that item $i\geq3$
is packed first. Since the smallest item has size $l(1)=F_{n}$, for
capacity $C_{i}=2F_{n}+F_{i}-2<2F_{n}+F_{i}-1=l(1)+l(i)$, no additional
item fits the knapsack. However, the unique optimum solution in this
case is $\textsc{Opt}(\mathcal{I},C_{i})=\{i-1,i-2\}$. These two
items fit the knapsack, as $l(i-1)+l(i-2)=2F_{n}+F_{i-1}+F_{i-2}-2=2F_{n}+F_{i}-2=C_{i}$.
By definition, 
\begin{align*}
\frac{v(i-1)+v(i-2)}{v_{i}}=\frac{2n+2i-3}{n+i}=2-\frac{3}{n+i}\geq2-\frac{3}{n}\,.
\end{align*}
Hence, policies that first pack item $i\geq3$ do not achieve a robustness
factor $\alpha < 2-3/n$.

Now, assume that one of the two smallest items is packed first. For
capacity $C_{1,2}=l(n)=2F_{n}-1<2F_{n}=l(1)+l(2)$, no additional
item fits the knapsack. The unique optimum solution, however, is to
pack item $n$. It remains to compute the ratios 
\[
\frac{v(n)}{v(1)}>\frac{v(n)}{v(2)}=\frac{2n}{n+2}=2-\frac{4}{n+2}>2-\frac{4}{n}\,.
\]
Hence, policies that first pack item 1 or item 2 do not achieve a
robustness factor $\alpha < 2-4/n$.  
\end{proof}

\section{Unit Densities\label{sec:Unit-Densities}}

In this section we restrict ourselves to instances of the oblivious
knapsack problem, where all items have unit density, i.e., $v(i)=l(i)$
for all items $i\in\I$. For two items $i,j\in\I$ we say that $i$
is smaller than $j$ and write $i\prec j$ if $v(i)<v(j)$, or $v(i)=v(j)$
and $t(i)<t(j)$, where $t$ is the tiebreaking index introduced in
Section~\ref{sec:Algorithm}. We give an algorithm $\textsc{UniversalUD}$~(cf.~Algorithm~\ref{our_algo_ud})
that produces a universal policy tailored to achieve the best possible
robustness factor equal to the golden ratio $\varphi\approx1.618$.
The algorithm considers the items from smallest to largest, and inserts
each item into the output sequence as far to the end as possible,
such that the item is not preceded by other items that are more than
a factor~$\varphi$ smaller. Intuitively, the algorithm tries as
much as possible to keep the resulting order sorted increasingly by
size; only when an item dominates another item by a factor of at least~$\varphi$
the algorithm ensures that it precedes this item in the final sequence.
Note that, even though $\varphi$ is irrational, for rationals $a,b$
the condition $a<\varphi b$ can be tested efficiently by testing
the equivalent condition $a/b<1+b/a$.

\begin{algorithm}[tb]      
  \DontPrintSemicolon    
  \KwIn{set of items $\I$}    
  \KwOut{sequence of items $\Pi$}   
  \caption{\textsc{UniversalUD}($\mathcal{I}$)\label{our_algo_ud}}        
  $L \leftarrow \left<\textrm{items in }\mathcal{I} \textrm{ sorted such that }L_1\prec\dots\prec L_n\right>$\;   
  $\Pi^{(0)} \leftarrow \emptyset$\;   
  \For{$r \leftarrow 1,\dots,n$}      
  {                
    $j \leftarrow 1$\;            
    \While{$j \leq |\Pi|$ {\bf and} $v(L_r) < \varphi v(\Pi^{(r-1)}_j)$}            
    { $j \leftarrow j+1$\; }            
    $\Pi^{(r)} \leftarrow (\Pi^{(r-1)}_1,\dots,\Pi^{(r-1)}_{j-1},L_r,\Pi^{(r-1)}_j,\dots)$\;        
  }   
  \Return $\Pi^{(n)}$\;  
\end{algorithm}
\begin{theorem}
The algorithm $\textsc{UniversalUD}$ constructs a universal policy
of robustness factor $\varphi$ when all items have unit density. \end{theorem}
\begin{proof}
Given an instance $\I$ of the oblivious knapsack problem with unit
densities and any capacity $C\leq v(\I)$, we compare the packing
$\Pi(C)$ that results from the solution $\Pi=\textsc{UniversalUD}(\I$)
with an optimal packing $\OPT(\I,C)$. We define the set $M$ of items
in $\Pi(C)$ for which at least one smaller item is not in $\Pi(C)$,
i.e., more precisely, let $M=\{i\in\Pi(C)\mid\exists j\in\I\backslash\Pi(C):j\prec i\}$.

We first consider the case that $M\neq\emptyset$ and set $i=\min_{\prec}M$
to be the smallest item in $M$ with respect to `$\prec$'. Consider
the iteration $r$ of $ $$\textsc{UniversalUD}$ in which $i$ is
inserted into $\Pi$, i.e., $i=L_{r}$. By definition of $M$, there
is an item $j\prec i$ with $j\notin\Pi(C).$ Let $j$ be the first
such item in $\Pi$. Since $j\prec i$, we have $j\in\Pi^{(r)}$.
From $i\in\Pi(C)$ and $j\notin\Pi(C)$, it follows that $i$ precedes
$j$ in $\Pi$ (and thus in $\Pi^{(r)}$). Let $i'$ be the item directly
preceding $j$ in $\Pi^{(r)}$. If $i'=i$, $i$ was compared with
$j$ when it was inserted into $\Pi^{(r)}$, with the result that
$v(i)\geq\varphi v(j)$ and thus $v(\Pi(C))\geq\varphi v(j)$. If
$i'\neq i$, by definition of $j$, we still have $i'\in\Pi(C)$.
Also, either $i'\succ j$ and thus $v(i')\geq v(j)$, or $j$ was
compared with $i'$ when it was inserted into $\Pi$ in an earlier
iteration of $\textsc{UniversalUD}$, with the result that $v(i')>\frac{1}{\varphi}v(j)$.
Again, $v(\Pi(C))\geq v(i)+v(i')>v(j)+\frac{1}{\varphi}v(j)=\varphi v(j)$.

In both cases it follows from $j\notin\Pi(C)$ that $v(\OPT(\I,C))\leq C<v(\Pi(C))+v(j)$,
and using $v(j)\leq\frac{1}{\varphi}v(\Pi(C))$ we get 
\[
\frac{v(\OPT(\I,C))}{v(\Pi(C))}<\frac{v(\Pi(C))+v(j)}{v(\Pi(C))}<1+\frac{1}{\varphi}=\varphi.
\]

Now, assume that $M=\emptyset$. Intuitively, this means that $\Pi(C)$
consists of a prefix of $L$ (the smallest items). Let $i_{1}\succ\dots\succ i_{k}$
be the items in $\Pi(C)\setminus\OPT(\I,C)$, and let $j_{1}\succ\dots\succ j_{l}$
be the items in $\OPT(\I,C)\setminus\Pi(C)$. As $\Pi(C)$ consists
of a prefix of $L$, we have $|\Pi(C)|\geq|\OPT(\I,C)|$ and thus
$k\geq l$. If $k=0$, the claim trivially holds. Otherwise, since
$M$ is empty, we have $j_{l}\succ i_{1}$. Is suffices to show $v(j_{h})\leq\varphi v(i_{h})$
for all $h\leq l$. To this end, we consider any fixed $h\leq l$.
From $v(\{i_{1},\dots,i_{h-1}\})\leq v(\{j_{1},\dots,j_{h-1}\})$
it follows that
\[
v(j_{h})\leq v(\OPT(\I,C))-v(\{j_{1},\dots,j_{h-1}\})\leq C-v(\{i_{1},\dots,i_{h-1}\}).
\]
 This implies that $j_{h}$ cannot precede all items of $\{i_{h},\dots,i_{k}\}$
in $\Pi$, as $j_{h}\notin\Pi(C)$. Hence, there is an item $i\in\{i_{h},\dots,i_{k}\}$
that precedes $j_{h}$ in $\Pi$. Since $j_{h}\succ i$, in the iteration
when $\textsc{UniversalUD}$ inserted $j_{h}$ into $\Pi$, $i$ was
already present. From the fact that $i$ ended up preceding $j_{h}$
it follows that $j_{k}$ was compared with $i$ and thus $v(j_{h})<\varphi v(i)\leq\varphi v(i_{h})$.
We obtain
\[
\frac{v(\OPT(\I,C))}{v(\Pi(C))}\leq\frac{v(\OPT(\I,C)\setminus\Pi(C))}{v(\Pi(C)\setminus\OPT(\I,C))}=\frac{\sum_{h=1}^{l}
v(j_{h})}{\sum_{h=1}^{k}v(i_{h})}\leq\frac{\sum_{h=1}^{l}\varphi v(i_{h})}{\sum_{h=1}^{l}v(i_{h})}=\varphi,
\]
{which implies the result.}
\end{proof}
A naïve implementation of \noun{UniversalUD }runs in time $\Theta(n^{2})$.
We improve this running time to $\Theta(n\log n)$. 
\begin{theorem}
The algorithm $\textsc{UniversalUD}$ can be implemented to run in
time $\Theta(n\log n)$. \label{thm:running time UniversalUD}\end{theorem}
\begin{proof}
To improve the running time from the naïve $\Theta(n^{2})$, we maintain
a balanced search tree $T$ that stores a subset of the items in $\Pi$
sorted decreasingly by their sizes. Whenever an item gets inserted
to the front of~$\Pi$, and only then, we also insert it into $T$.
This way, the items in $T$ remain sorted by their positions in $\Pi$
throughout the execution of the algorithm. We need an efficient way
of finding, in each iteration $r$ of $\textsc{UniversalUD}$~(Algorithm~\ref{our_algo_ud}),
the first item $i$ in $\Pi^{(r)}$ for which $v(L_{r})\geq\varphi v(i)$,
or detecting that no such item exists. We claim that, if such an item
exists, it is stored in $T$ and can thus be found in time $\Theta(\log n)$.

It suffices to show that for every item $i\in T$ and its predecessor
$j$ in $T$ we have that none of the items that precede $i$ in $\Pi$
are smaller than $j$. To see this, we argue that none of the items
between $j$ and $i$ in $\Pi$ are smaller than $j$. We can then
repeat the argument for $j$ and its predecessor $j'$, etc. For the
sake of contradiction, let $i'$ be the first item between $j$ and
$i$ with $v(i')<v(j)$. None of the items between $j$ and $i'$
are smaller than $j$, hence both $j$ and $i'$ are inserted into
$\Pi$ earlier than all of them. Let $r$ be the iteration in which
$j$ is inserted into $\Pi$. Since $i'$ is inserted earlier into
$\Pi$, and since $j$ is inserted to the front of $\Pi^{(r)}$, $i'$
is at the front of $\Pi^{(r-1)}$. This is a contradiction to $i'$
not being in $T$.
\end{proof}
We now establish that $\textsc{UniversalUD}$ is best possible, even
if we permit non-universal policies.
\begin{theorem}
There are instances of the oblivious knapsack problem where no policy
achieves a robustness factor of $\varphi-\delta$, for any $\delta>0$,
even when all items have unit density.\end{theorem}
\begin{proof}
Consider an instance of the oblivious knapsack problem with five items
of unit density and values equal to
$v_{1}=1+\varepsilon,v_{2}=1+\varepsilon,v_{3}=2/\varphi,v_{4}=1+1/\varphi^{2},v_{5}=\varphi$,
for sufficiently small $\varepsilon>0$. We show that no algorithm
achieves a robustness factor of $\varphi-\delta$ for this instance.
To this end we consider an arbitrary algorithm $\mathscr{A}$ and
distinguish different cases depending on which item the algorithm
tries to pack first. 
\renewcommand{\theenumi}{(\alph{enumi})}
\renewcommand{\labelenumi}{\theenumi}

\begin{enumerate}
\item If $\mathscr{A}$ tries item 1 or item 2 first, it cannot fit any
additional item for a capacity equal to $v_{5}=\varphi$, as even
$v_{1}+v_{2}>\varphi$. For this capacity $\mathscr{A}$ is worse
by a factor of $\varphi/(1+\varepsilon)>\varphi-\delta$ than the
optimum solution, which packs item 5. 

\item If $\mathscr{A}$ tries item 3 first, it cannot fit any additional
item for a capacity equal to $v_{1}+v_{2}=2+2\varepsilon$, as even
$v_{3}+v_{1}>2+2\varepsilon$. For this capacity $\mathscr{A}$ is
worse by a factor of $(1+\varepsilon)\varphi>\varphi-\delta$ than
the optimum solution which packs items 1 and~2. 
\item If $\mathscr{A}$ tries item 4 first, it cannot fit any additional
item for a capacity equal to $v_{2}+v_{3}=1+2/\varphi+\varepsilon$,
as even $v_{4}+v_{1}=2+1/\varphi^{2}+\varepsilon>1+2/\varphi+\varepsilon$.
For this capacity $\mathscr{A}$ is worse by a factor of
$\frac{1+2/\varphi+\varepsilon}{1+1/\varphi^{2}}>\frac{\varphi+1/\varphi}{1+1/\varphi^{2}}=\varphi>\varphi-\delta$
than the optimum solution which packs items 2 and 3. 
\item If $\mathscr{A}$ tries item 5 first, it cannot fit any additional
item for a capacity equal to $v_{3}+v_{4}=\varphi+1$, as even $v_{5}+v_{1}=\varphi+1+\varepsilon>\varphi+1$.
For this capacity $\mathscr{A}$ is worse by a factor of $\frac{\varphi+1}{\varphi}=\varphi>\varphi-\delta$
than the optimum solution which packs items 3 and 4.  
\end{enumerate}
\end{proof}

\section{Hardness\label{sec:Hardness}}

Although we can always find a $2$-robust universal policy in polynomial
time, we show in this section that, for any fixed $\alpha\in(1,\infty)$,
it is intractable to decide whether a given policy is $\alpha$-robust,
even if it is universal. This hardness result also holds for instances
with unit densities when $\alpha$ is part of the input. As the final
-- and arguably the most interesting -- result of this section, we
establish $\classcoNP$-hardness of the the problem to decide for
a given instance and given $\alpha>1$, whether the instance admits
a universal policy with robustness factor $\alpha$. All proofs rely
on the hardness of the following version of $\textsc{SubsetSum}$.
\begin{lemma}
\label{lem:subset_sum}Let $W=\{w_{1},w_{2},\dots,w_{n}\}$ be a set
of positive integer weights and $T\leq\sum_{k=1}^{n}w_{k}$ be a target
sum. The problem of deciding whether there is a subset $U\subseteq W$
with $\sum_{w\in U}w=T$ is $\classNP$-complete, even when\end{lemma}
\begin{enumerate}
\item $T=2^{k}$ for some integer $k\geq3$,\label{enu:subsetSum1}

\begin{enumerate}
\item all weights are in the interval $[2,T/2)$,\label{enu:subsetSum2}
\item all weights have a difference of at least 2 to the closest power of
2.\label{enu:subsetSum3}
\end{enumerate}
\end{enumerate}
\begin{proof}
Without Properties~\ref{enu:subsetSum1} to \ref{enu:subsetSum3},
the $\textsc{SubsetSum}$ problem is well known to be $\classNP$-complete
(e.g., Garey and Johnson~\cite{gareyJ79}). Given an instance $(W,T)$
of this classical problem, we construct an equivalent instance with
Properties~\ref{enu:subsetSum1} to \ref{enu:subsetSum3}. We first
multiply all weights in $W$ as well as the target sum $T$ with $6$
to obtain an equivalent instance $(W',T')$. In the new instance,
all weights are even but not a power of 2, hence they have distance
at least 2 to the closest power of 2. We set $T''=2^{\sigma}$, with
$\sigma=\left\lceil \log_{2}(T'+\sum_{w'\in W'}w')\right\rceil +2$
and define two new weights
\[
u=\left\lfloor \frac{T''-T'}{2}\right\rfloor ,\quad w=\left\lceil \frac{T''-T'}{2}\right\rceil .
\]
We set $W''=W'\cup\{u,w\}$ to obtain the final instance $(W'',T'')$.
Properties~\ref{enu:subsetSum1} and \ref{enu:subsetSum2} are satisfied
by construction. Also, any solution to the instance $(W'',T'')$ has
to include both $u$ and $w$, since $T''>4\cdot\sum_{w'\in W'}w'$.
Hence, the instance remains equivalent to the original instance $(W,T)$.
Since $T''-T'>3T''/4$, and since $T''$ is a power of two, the new
items $u$ and $w$ are far enough from the closest power of 2 (which
either is $T''/2$ or $T''/4$).
\end{proof}
We first show that it is intractable to determine the robustness factor
of a given universal policy.
\begin{theorem}
\label{thm:hardness} For any fixed and polynomially representable
$\alpha>1$ it is $\classcoNP$-complete to decide whether a given
universal policy for the oblivious knapsack problem is $\alpha$-robust.\end{theorem}
\begin{proof}
Regarding the membership in $\classcoNP$, note that if a universal
policy $\Pi$ is not $\alpha$-robust, then there is a capacity $C$
such that $v(\Pi(C))<v(\OPT(\I,C))/\alpha$. Thus, $C$ together with
$\OPT(\I,C)$ is a certificate for $\Pi$ not being an $\alpha$-robust
solution. 

For the proof of $\classcoNP$-hardness, we reduce from the variant
of \textsc{SubsetSum} specified in Lemma~\ref{lem:subset_sum}. An
instance of this problem is given by a set $W=\{w_{1},w_{2},\dots,w_{n}\}$
of positive integer weights in the range $[2,T/2)$ and a target sum
$T=2^{k}$ for some integer $k\geq3$. Let $\alpha>1$ be polynomially
representable. We may assume without loss of generality that $\alpha>\frac{T}{T-1}$
as we can ensure this property by multiplying $T$ and all items in
$W$ by a sufficiently large power of 2.

We construct an instance $\I$ and a sequence $\Pi$ such that $\Pi$
is an $\alpha$-robust universal policy for $\mathcal{I}$ if and
only if the instance of \textsc{SubsetSum} given by $W$ and $T$
has no solution. To this end, we introduce for each weight $w\in W$
an item with value and size equal to $w$. In this way, the optimal
knapsack solution for capacity $T$ is at least $T$ if the instance
of \textsc{SubsetSum} has a solution. Furthermore, we introduce a
set of additional items that make sure that the robustness factor
for all capacities except $T$ is at most $\alpha$ while maintaining
the property that the optimal knapsack solution for capacity $T$
is strictly less than $T$ if the instance of \textsc{SubsetSum} has
no solution.

We now explain the construction of $\I$ and $\Pi$ is detail. Let
$\smash{\epsilon=\frac{\alpha(T-1)-T}{\alpha(T-1)-1}}$, i.e., $\smash{\alpha=\frac{T-\epsilon}{(T-1)(1-\epsilon)}}$.
Note that $\epsilon\in(0,1)$ by our assumptions on $T$ and $\alpha$.
For each weight $w\in W$, we introduce an item $i_{w}$ with $l(i_{w})=v(i_{w})=w$.
The set of these items is called \emph{regular} and is denoted by
$\I_{\text{{reg}}}$. Furthermore, we introduce a set of auxiliary
items. Let $m=\log_{2}T-1$. Then, for each $k\in\{0,1,\dots,m\}$,
we introduce an auxiliary item $j_{k}$ with size $l(j_{k})=2^{k}$
and value $v(j_{k})=2^{k}\,(1-\epsilon)$. Denoting the set of auxiliary
items by $\I_{\text{{aux}}}$, we have $l(\I_{\text{{aux}}})=\sum_{k=0}^{m}2^{k}=T-1$.
Finally, we introduce a dummy item~$d$ with $l(d)=T+1$ and 
\[
v(d)=\frac{1-\epsilon}{\epsilon}\left(v(\I_{\mathrm{aux}})+v(\I_{\mathrm{reg}})\right)=\frac{1-\epsilon}{\epsilon}
\Biggl((T-1)(1-\epsilon)+\sum_{w\in W}w\Biggr).
\]
The universal policy $\Pi$ is defined as $\Pi=(d,j_{m},j_{m-1},\dots,j_{0},i_{w_{n}},i_{w_{n-1}},\dots,i_{w_{1}})$.
The hardness proof relies on the claim that $\Pi$ is a $\frac{1}{1-\epsilon}$-robust
universal policy for all capacities except $T$, i.e., 
\begin{equation}
v(\OPT(\I,C))\leq\frac{1}{1-\varepsilon}v(\Pi(C))\,\textrm{for all }C\neq T.\label{eq:except at T}
\end{equation}

As all item sizes are integer, it suffices to consider integer capacities.
To prove \eqref{eq:except at T}, let us first consider capacities
$C\leq T-1$. Since the density of each item with size not larger
than $T-1$ is bounded from above by $1$, it is sufficient to show
that $v(\Pi(C))=C(1-\epsilon)$. To this end, we show that every capacity
$C\in\{1,\dots,2^{m+1}-1=T-1\}$ is packed without a gap by the exponentially
decreasing sequence of items $j_{m},j_{m-1},\dots,j_{0}$. We prove
this statement by induction over $m$. For~$m=0$, the statement
is true, since there is only a single item with length $1$, which
packs the capacity $C=1$ optimally. Now assume that the statement
is true for all $m'<m$ and consider the sequence $j_{m},j_{m-1},\dots,j_{0}$.
We distinguish two cases. For capacities $C\in\{2^{m},\dots,2^{m+1}-1\}$,
item~$j_{m}$ is packed and, using the induction hypothesis, the
residual capacity $\tilde{C}=C-2^{m}\leq2^{m+1}-1-2^{m}\leq2^{m}-1$
can be packed without a gap by the remaining sequence $j_{m-1},j_{m-2},\dots,j_{0}$.
For capacities $C<2^{m}$, item $j_{m}$ is not packed, and, again
using the induction hypothesis, we derive that $C$ can be packed
by $j_{m-1},\dots,j_{0}$. This completes the proof of our claim for
$C\leq T-1$.

Let us now consider our claim for capacities $C\geq T+1$. In this
case, $d\in\Pi(C)$ and we can trivially bound the robustness factor
of $\Pi$ by observing that
\[
\frac{v(\OPT(\I,C))}{v(\Pi(C))}\leq\frac{v(\I)}{v(d)}=1+\frac{(T-1)(1-\epsilon)+\sum_{w\in
W}w}{v(d)}=1+\frac{\epsilon}{1-\epsilon}=\frac{1}{1-\epsilon}.
\]

We proceed to show that $\Pi$ is an $\alpha$-robust universal policy
if and only if the instance of \textsc{SubsetSum} given by $W$ and
$T$ has no solution. Let us first assume that the instance of \textsc{SubsetSum}
has no solution. We prove that $\Pi$ is $\alpha$-robust. For all
capacities except $T$ this is clear from claim~ \eqref{eq:except at T}.
For capacity $T$, we argue as follows: As there is no packing of
$T$ with items of density~1, we bound $v(\OPT(\I,T))$ from above
by $(T-1)+(1-\epsilon)$, whereas $\Pi$ packs all auxiliary items.
We get 
\[
\frac{v(\OPT(\I,T))}{v(\Pi(T))}\leq\frac{(T-1)+(1-\epsilon)}{(T-1)(1-\epsilon)}=\alpha.
\]

Now, assume that the instance of \textsc{SubsetSum} has a solution.
Then, $v(\OPT(T))=T$ and thus 
\begin{align*}
\frac{v(\OPT(\I,T))}{v(\Pi(T))}=\frac{T}{(T-1)(1-\epsilon)}>\alpha,
\end{align*}
and we conclude that $\Pi$ is not $\alpha$-robust. 
\end{proof}

We give a result similar to Theorem~\ref{thm:hardness} for instances in which each item has unit
density. Note that this time we require $\alpha$ to be part of
the input.

\pagebreak

\begin{theorem}
\label{thm:hardness_unit} It is $\classcoNP$-complete to decide
whether, for given $\alpha>1$, a given universal policy for the oblivious
knapsack problem is $\alpha$-robust, even when all items have unit
density.\end{theorem}
\begin{proof}
Membership in $\classcoNP$ follows from Theorem~\ref{thm:hardness}.
To prove hardness, we again reduce from \textsc{SubsetSum} (Lemma~\ref{lem:subset_sum})
using a similar construction as in the proof of Theorem~\ref{thm:hardness}.
Let the set $W=\{w_{1},\dots,w_{n}\}$ of weights and the target sum
$T\geq8$ of an instance of \textsc{SubsetSum} be given, with $w_{1}\leq w_{2}\leq\dots\leq w_{n}$.
We proceed to explain the construction of a universal policy $\Pi$
for which the decision whether $\Pi$ is $\alpha$-robust is $\classcoNP$-hard,
for some $\alpha>1$.

For each weight $w\in W$, we introduce an item $i_{w}$ with value
$v(i_{w})=w$. The set of these items is called \emph{regular} and
is denoted by $\I_{\text{{reg}}}$. Let $m=\log_{2}T-1$ and $\epsilon=1/T^{2}$.
For each $k\in\{0,\dots,m\}$, we introduce an auxiliary item $j_{k}$
with value $v(j_{k})=2^{k}(1-\epsilon)$. Denoting the set of auxiliary
items by $\I_{\text{{aux}}}$, we have $v(\I_{\text{{aux}}})=(1-\epsilon)\sum_{k=0}^{m}2^{k}=(1-\epsilon)(T-1)$.
We further introduce a set of dummy items $\I_{\text{{dum}}}=\{d_{0},\dots,d_{m'}\}$,
where $m'=\lceil\log_{2}w_{n}\rceil$. We set $v(d_{k})=T\cdot2^{k}$
for each $k\in\{1,\dots,m'\}$, and $v(d_{0})=T+\varepsilon$. The
values of the dummy items sum up to $v(\I_{\text{{dum}}})=(T+\epsilon)+T\sum_{k=1}^{m'}2^{k}=T(2^{m'+1}-1)+\epsilon$.
In total, the sum of the values of all dummy and auxiliary items is
\begin{align}
S=v(\I_{\text{{aux}}})+v(\I_{\text{{dum}}})=(1-\epsilon)(T-1)+T(2^{m'+1}-1)+\epsilon.\label{eq:sum-1}
\end{align}
Finally, we define the sequence $\Pi$ as 
\begin{align*}
\Pi=(d_{m'},d_{m'-1},\dots,d_{0},j_{m},j_{m-1},\dots,j_{0},i_{w_{n}},i_{w_{n-1}},\dots,i_{w_{1}}),
\end{align*}
i.e., $\Pi$ first tries to pack the dummy items in decreasing order,
then the auxiliary items in decreasing order, and finally the regular
items in non-increasing order. Let $\alpha=\frac{T-\epsilon}{(1-\epsilon)(T-1)}$.
We proceed to prove the statement of the theorem by showing that $\Pi$
is an $\alpha$-robust universal policy if and only if the instance
$(W,T)$ of $\textsc{SubsetSum}$ has no solution. To this end, we
first prove that $\Pi$ is always an $\alpha$-robust universal policy
for all capacities except the \emph{critical} capacities in the interval
$\big[T-\epsilon T,T+\epsilon\big)$. Then, we argue that $\Pi$ is
$\alpha$-robust for the critical capacities if and only if the instance
$(W,T)$ of \textsc{SubsetSum} has no solution.

We start by proving that $v(\Pi(C))$ is within an $\alpha$-fraction
of $v(\OPT(C))$ for all capacities $C\in[0,T-\epsilon T)$. Since
the regular items are of integer values and the values of the auxiliary
items each are an $(1-\epsilon)$-fraction of an integer, only capacities
$C$ for which the ratio $C/\lceil C\rceil$ is not smaller than $1-\epsilon$
can be packed without a gap. Otherwise, the value of an optimal solution
is bounded from above by $\lfloor C\rfloor$. For capacities $C\in[0,T-\epsilon T)$,
we obtain
\begin{equation}
v(\OPT(\I,C))\leq\begin{cases}
C, & \text{ if }C/\lceil C\rceil\geq1-\epsilon\\
\lfloor C\rfloor, & \text{ otherwise.}
\end{cases}\label{eq:bound_opt-1}
\end{equation}
The value packed by $\Pi$ is given by
\begin{equation}
v(\Pi(C))=\begin{cases}
(1-\epsilon)\lceil C\rceil, & \text{ if }C/\lceil C\rceil\geq1-\epsilon\\
(1-\epsilon)\lfloor C\rfloor, & \text{ otherwise.}
\end{cases}\label{eq:bound_us-1}
\end{equation}
From \eqref{eq:bound_opt-1} and \eqref{eq:bound_us-1} it follows
that 
\begin{align}
v(\OPT(\I,C))\leq\frac{1}{1-\epsilon}v(\Pi(C))<\alpha\, v(\Pi(C))\label{eq:bound_first_case-1}
\end{align}
for all $C\in[0,T-\epsilon T)$.

We proceed to prove that $\Pi$ is within an $\alpha$-fraction of an optimal
solution for all capacities $C\in[T+\epsilon,S]$. We distinguish
two cases for each such capacity $C$.

\emph{First case:} $\I_{\text{{aux}}}\subset\Pi(C)$, i.e., all auxiliary
items are packed by $\Pi$. Since, in $\Pi$, the dummy item $d_{0}$
with value $T+\epsilon$ precedes all auxiliary items, and since $C\geq T+\epsilon$,
this case can only occur for capacities
\begin{equation}
C\geq v(d_{0})+v(\I_{\text{\text{{aux}}}})=T+\varepsilon+(1-\varepsilon)(T-1)=2(T+\varepsilon)-(1+\varepsilon
T).\label{eq:min_C-1}
\end{equation}
On the other hand, the gap $C-v(\Pi(C))$ is at most the gap left
after trying all dummy items and packing all auxiliary items, i.e.,
$C-v(\Pi(C))<v(d_{0})-v(\I_{\text{{aux}}})=T+\epsilon-(1-\epsilon)(T-1)=1+\epsilon T$.
Thus,
\begin{multline*}
\frac{v(\OPT(\I,C))}{v(\Pi(C))}<\frac{C}{C-(1+\epsilon
T)}\overset{\eqref{eq:min_C-1}}{\leq}\frac{2(T+\varepsilon)-(1+\varepsilon T)}{2(T+\varepsilon)-2(1+\varepsilon T)}\\
=\frac{(T+\epsilon)-(1+\epsilon T)/2}{(T+\varepsilon)-(1+\varepsilon
T)}\overset{T\geq8}{<}\frac{T-\epsilon}{(1-\epsilon)(T-1)}=\alpha.
\end{multline*}

\emph{Second case: $\I_{\text{{aux}}}\setminus\Pi(C)\neq\emptyset$},
i.e., not all auxiliary items are packed. This implies that the gap
$C-v(\Pi(C))$ is at most $1-\epsilon$. We calculate
\[
\frac{v(\OPT(\I,C))}{v(\Pi(C))}<\frac{C}{C-(1-\epsilon)}\overset{C\geq
T+\varepsilon}{\leq}\frac{T+\epsilon}{T+2\epsilon-1}\overset{\varepsilon=1/T^{2}}{<}\frac{T-\epsilon}{(1-\epsilon)(T-1)}
=\alpha.
\]

Next, we consider capacities $C\in(S,v(\I_{\text{{aux}}}\cup\I_{\text{{dum}}}\cup\I_{\text{{reg}}})]$.
For these capacities, all dummy items and all auxiliary items are
packed by $\Pi$. Using that the gap $C-\Pi(C)$ is at most $w_{n}$,
we obtain
\begin{multline*}
\frac{v(\OPT(\I,C))}{v(\Pi(C))}\leq\frac{C}{C-w_{n}}\overset{C>S}{<}\frac{S}{S-w_{n}}\overset{S>T2^{m'}}{<}\frac{T2^{m'}}{
T2^{m'}-w_{n}}\\
\leq\frac{Tw_{n}}{Tw_{n}-w_{n}}=\frac{T}{T-1}=\frac{T(1-\varepsilon)}{(1-\epsilon)(T-1)}<\frac{T-\epsilon}{(1-\epsilon)(T-1)
}=\alpha.
\end{multline*}

To finish the proof, let us finally consider the critical capacities
$C\in\big[T-T\epsilon,T+\epsilon\big)$. We proceed to show that $v(\Pi(C))$
is within an $\alpha$-fraction of $v(\OPT(C))$ for all $C\in\big[T-T\epsilon,T+\epsilon\big)$
if and only if $(W,T)$ does not have a solution. Let us first assume
that $(W,T)$ does not have a solution. Then, $v(\OPT(C))\leq T-\epsilon$
and we obtain
\[
\frac{v(\OPT(\I,C))}{v(\Pi(C))}\leq\frac{T-\epsilon}{(T-1)(1-\epsilon)}=\alpha,
\]
for all $C\in\big[T-T\epsilon,T+\epsilon\big)$. If, on the other
hand, $(W,T)$ has a solution, then $v(\OPT(T))=T$, implying that 

\[
\frac{v(\OPT(\I,T))}{v(\Pi(T))}=\frac{T}{(T-1)(1-\epsilon)}>\alpha,
\]
i.e., $\Pi$ is not an $\alpha$-robust universal policy.
\end{proof}

Finally, we prove that it is hard to decide whether a given instance
admits an $\alpha$-robust universal policy when $\alpha$ is part
of the input.

\pagebreak 
\begin{theorem}
It is $\classcoNP$-hard to decide whether, for given $\alpha>1$,
an instance of the oblivious knapsack problem admits an $\alpha$-robust
universal policy, even when all items have unit density.\end{theorem}
\begin{proof}
We again reduce from \textsc{SubsetSum}. To this end, let $(W,T)$
be an instance of \textsc{SubsetSum}~(Lemma~\ref{lem:subset_sum}),
let $\I$ be the set of items constructed from $(W,T)$ in the proof
of Theorem~\ref{thm:hardness_unit}, and let $\alpha=\frac{T-\epsilon}{(1-\epsilon)(T-1)}$.
We proceed to show that $\I$ admits an $\alpha$-robust universal
policy if and only if the instance \textsc{$(W,T)$ }of \textsc{SubsetSum}
has no solution.

For the case that $(W,T)$ has no solution, an $\alpha$-robust universal
policy is constructed in the proof of Theorem~\ref{thm:hardness_unit}.
Thus, it suffices to show that if $(W,T)$ has a solution, $\I$ does
not admit an $\alpha$-robust universal policy.

First, we claim that any $\alpha$-robust universal policy $\Pi$
contains the auxiliary items in decreasing order. Otherwise, for the
sake of contradiction, let $j$ be the first auxiliary item in $\Pi$
that is preceded by a smaller auxiliary item $i$. Consider the capacity
$C=v(j)$. As all dummy items are larger than $T>C$, only auxiliary
and regular items can be in $\Pi(C)$. Since $i$ precedes $j$, we
have $j\notin\Pi(C)$. 

If\emph{ }$\Pi(C)$ contains only auxiliary items, since the sum of
the values of the auxiliary items smaller than $v(j)$ is $v(j)-(1-\epsilon)$,
we can use that $j\notin\Pi(C)$ to obtain $v(\Pi(C))\leq v(j)-(1-\epsilon)<\left\lfloor v(j)\right\rfloor $.
If \emph{$\Pi(C)$} contains a regular item $i'$, then $\frac{C-v(i')}{\left\lceil C-v(i')\right\rceil }<1-\varepsilon$,
and hence the gap $C-v(i')$ cannot be packed with a value more than
$\left\lfloor C-v(i')\right\rfloor $. It follows that $v(\Pi(C))\leq\lfloor v(j)\rfloor$.
In either case we have
\[
\frac{v(\OPT(\I,C))}{v(\Pi(C))}\geq\frac{v(j)}{\lfloor
v(j)\rfloor}\overset{v(j)\leq(1-\varepsilon)T/2}{\geq}\frac{(1-\varepsilon)T/2}{\lfloor(1-\varepsilon)T/2\rfloor}=\frac{
(1-\varepsilon)T/2}{T/2-1}\overset{\varepsilon=1/T^{2}}{>}\frac{T-\epsilon}{(T-1)(1-\epsilon)}=\alpha.
\]
This is a contradiction to the assumption that $\Pi$ is $\alpha$-robust.
We conclude that the auxiliary items appear in $\Pi$ in decreasing
order.

Second, we claim that if $\Pi(T)$ contains a regular item, then $\Pi$
is not $\alpha$-robust. By the argument above, we may assume that
the auxiliary items in $\Pi$ are ordered decreasingly. Let $i$ be
the regular item contained in $\Pi(T)$ that appears first in $\Pi$.
Consider the capacity $C=(v(i)+1)(1-\epsilon).$ The auxiliary items
that appear before $i$ in $\Pi$ (if any) are ordered decreasingly.
All of them must be larger than $v(i)$, otherwise, the gap left after
packing them for capacity $T$ would be too small to fit $i$. By
Lemma~\ref{lem:subset_sum}, we have that neither $v(i)$ nor $v(i)+1$
are a power of 2, thus $\Pi(C)$ does not contain any of the auxiliary
items preceding $i$. All regular items that appear before $i$ in
$\Pi$ are larger than $v(i)$, since they are not in $\Pi(T)$. Hence,
$\Pi(C)$ does not contain any regular items except $i$. We conclude
that $\Pi(C)=\{i\}$. On the other hand, $C$ is an integer multiple
of $1-\epsilon$ and can be packed without a gap by auxiliary items
only. We obtain
\[
\frac{v(\OPT(C))}{v(\Pi(C))}=\frac{C}{v(i)}=\frac{(v(i)+1)(1-\varepsilon)}{v(i)}\overset{v(i)\leq
T/2}{\geq}\frac{(T/2+1)(1-\varepsilon)}{T/2}\overset{\varepsilon=1/T^{2}}{>}\alpha.
\]
We conclude that if an $\alpha$-robust universal policy $\Pi$ exists,
then $\Pi(T)$ does not contain regular items. It follows that $\Pi(T)=\I_{\text{{aux}}}$
and, thus, $v(\Pi(T))=(T-1)(1-\epsilon).$ Using that the $\textsc{SubsetSum}$
instance $(W,T)$ has a solution, we obtain

\[
\frac{v(\OPT(\I,T))}{v(\Pi(T))}\geq\frac{T}{(T-1)(1-\epsilon)}>\alpha,
\]
which implies that no $\alpha$-robust universal policy exists.
\end{proof}
\bibliographystyle{abbrv}
\bibliography{univ-knapsack}

\end{document}